\documentclass[final,1p,times,12pt]{elsarticle}

\journal{\,}

\usepackage{amsfonts,setspace}
\usepackage{amsthm}
\usepackage{amsmath}
\usepackage{amssymb}
\usepackage{color}
\usepackage{graphicx}
\usepackage{float}
\usepackage{mathrsfs}
\usepackage{mathtools}
\usepackage{latexsym}
\usepackage{subfig}
\usepackage{bm} 

\usepackage{geometry}
\usepackage{algorithm}
\usepackage[noend]{algpseudocode}

\usepackage{natbib}
\usepackage{flexisym}

\usepackage{lineno,hyperref}
\modulolinenumbers[5]

\hypersetup{
    colorlinks,
    citecolor=blue,
    filecolor=magenta,
    linkcolor=red,
}

\newtheorem{Lemma}{Lemma}

\def\AA{\mathcal{A}}
\def\HH{\mathcal{H}}
\def\LL{\mathcal{L}}
\def\KK{\mathcal{K}}

\def\RR{\mathcal{R}}

\def\xx{\text{\bf x}}
\def\cc{\text{\bf c}}
\def\yy{\text{\bf y}}

\def\ro{r_{\rm o}}

\begin{document}

\begin{frontmatter}




\title{\textbf{High order surface radiation conditions for time--harmonic waves in exterior domains}}

\author[label]{Sebasti\'{a}n Acosta}
\ead{sebastian.acosta@bcm.edu}
\ead[url]{https://sites.google.com/site/acostasebastian01}

\address[label]{Department of Pediatrics, Baylor College of Medicine and Texas Children's Hospital, TX}

\begin{abstract}
We formulate a new family of high order on--surface radiation conditions to approximate the outgoing solution to the Helmholtz equation in exterior domains. Motivated by the pseudo--differential expansion of the Dirichlet--to--Neumann operator developed by Antoine et al. (J. Math. Anal. Appl. 229:184--211, 1999), we design a systematic procedure to apply pseudo--differential symbols of arbitrarily high order. Numerical results are presented to illustrate the performance of the proposed method for solving both the Dirichlet and the Neumann boundary value problems. Possible improvements and extensions are also discussed.
\end{abstract}


%35J05  	Laplacian operator, reduced wave equation (Helmholtz equation)
%58J40  	Pseudodifferential and Fourier integral operators on manifolds
%58J05  	Elliptic equations on manifolds, general theory
%41A60  	Asymptotic approximations, asymptotic expansions
%65N38  	Boundary element methods

%Highlights:
%- A new family of high order on-surface radiation conditions is formulated.
%- Arbitrarily high order can be obtained through a recursive formula to estimate the DtN map.
%- Numerical examples show a systematic reduction in the error as the order increases.

\end{frontmatter}

%\linenumbers

%%%%%%%%%%%%%%%%%%%%%%%%%%%%%%%%%%%%%%%%%%%%%
%%%%%%%%%%% NEW SECTION %%%%%%%%%%%%%%%%%%%%%
%%%%%%%%%%%%%%%%%%%%%%%%%%%%%%%%%%%%%%%%%%%%%
\section{Introduction} \label{Section.Intro}

This paper is concerned with the approximation of the outgoing solution to the Helmholtz equation in a domain $\Omega$ exterior to a closed surface $\Gamma$. We consider both the Dirichlet and the Neumann boundary value problems,
\begin{subequations} \label{Eqn.Main}
\begin{align} 
& \Delta u + k^2 u = 0 \qquad \text{in $\Omega$}, \label{Eqn.Main1} \\
& u = f_{\rm Dir} \quad \text{or} \quad \partial_{\nu} u = f_{\rm Neu} \qquad \text{on $\Gamma$}, \label{Eqn.Main2} \\
& \lim_{r \to \infty} r \left( \partial_{r}u - ik u \right) = 0,  \label{Eqn.Main3}
\end{align}
\end{subequations}
where $r = |\xx|$ and $\xx \in \mathbb{R}^3$. The main underlying assumptions are that $k > 0$ is constant in $\Omega$, that $\Gamma$ is sufficiently smooth, and that the domain enclosed by $\Gamma$ is strictly convex.
These assumptions imply that the wave field in $\Omega$ is purely outgoing, ruling out the possibility of waves reflecting back and forth between disjoint subsets of $\Gamma$.

When exact analytical solutions to the boundary value problem (\ref{Eqn.Main}) are not available, boundary integral equations and their discretizations offer some of the most robust choices to approximate the solution numerically \cite{Kress-Book-1999,McLean2000,Martin-Book-2006,Chen2010-Book}. Another approach is to truncate and discretize the domain surrounding $\Gamma$ and impose an absorbing condition on an artificial boundary to approximate the outgoing nature of the solution. Some reviews of these techniques are found in \cite{Rabinovich2010,Givoli2004,Tsynkov1998,Ihlenburg1998-Book,Zarmi2013}. A third line of work is the development of \textit{on--surface radiation conditions} (OSRC). These conditions can be seen as the application of absorbing boundary conditions directly on the surface $\Gamma$ in order to approximate the Neumann data when the Dirichlet data is prescribed, or vice versa. Once the Cauchy data $(u, \partial_{\nu}u)$ on $\Gamma$ is suitably approximated, then the Green's representation
\begin{align} \label{Eqn.Green}
u(\xx) = \int_{\Gamma} \left[ u(\yy) \partial_{\nu(\yy)} \Phi(\xx,\yy) - \partial_{\nu} u(\yy) \Phi(\xx,\yy) \right] dS(\yy), \qquad \xx \in \Omega,
\end{align}
provides an approximate solution in $\Omega$. Here $\Phi$ is the fundamental solution to the Helmholtz equation. The OSRC is a semi--analytical approach meant to sacrifice some accuracy in favor of tremendous computational speed. The OSRC may also render ultra fast methods to explore parameter spaces for optimization problems, for the construction of reduced bases, or inexpensive preconditioners to solve boundary integral equations numerically via Krylov iterative techniques \cite{Antoine2005,Antoine2007,Darbas2013}. In general, the OSRC provides a crude approximation of the solution. The state--of--the--art lacks error estimates, and the only notion of convergence is obtained in a microlocal sense. Therefore, the OSRC should not be employed if a high degree of accuracy is needed. 

The original conception of the OSRC is due to Kriegsmann et al. \cite{Kriegsmann1987}. Over the years, improvements and extensions have been devised, including the works of Jones \cite{Jones1988,Jones1990,Jones1992}, Ammari \cite{Ammari1998,Ammari1998b}, Calvo et al. \cite{Calvo2004,Calvo2003}, Antoine, Barucq et al. \cite{Antoine1999,Antoine2008,Antoine2001,Antoine2006,
Barucq2003,Barucq2010a,Barucq2012,Alzubaidi2016} and Darbas et al. \cite{Antoine2004,Antoine2005,Antoine2006,Antoine2007,Darbas2006,Darbas2013,
Darbas2015,Chaillat2015}. See also \cite{Atle2007,Acosta2015f,Stupfel1994,Murch1993,Teymur1996,Yilmaz2007,
Medvinsky2010,Chniti2016,Alzubaidi2016}. An overview of the OSRC theory, with an extensive list of references, was published in \cite{Antoine2008}. One of the open problems discussed by Antoine therein was the lack of a systematic procedure to increase the order of the OSRC beyond second order. Our present work is a novel attempt to partially solve this problem. Other efforts have been carried out in the past. Teymur extended the formulation of Jones, which assumes that the wave field has a well defined phase front conforming to the surface $\Gamma$, to furnish OSRC of high order \cite{Teymur1996}. Explicit expressions were provided up to fourth order. Calvo \cite{Calvo2003,Calvo2004} started from the eigenfunction representation of the Dirichlet--to--Neumann map for a circle and used rational approximations for the logarithmic derivative of the Hankel functions to obtained applicable OSRC of high order. 

As reviewed in Section \ref{Section.PDO}, the approach developed in the present paper is primarily influenced by \cite{Antoine1999} which is based on rigorous microlocal analysis to express the outgoing Dirichlet--to--Neumann (DtN) map in a classical pseudo--differential sense. This pseudo--differential calculus renders a recursive formula to obtain a differential expansion of the DtN map to an arbitrarily high order. In Section  \ref{Section.SphericalCase}, we show that the recursive formula becomes algorithmically tractable for the spherical geometry which we formally generalize to other arbitrary convex surfaces. In Section \ref{Section.Numerics}, we propose a numerical implementation of the high order OSRC based on a triangulation for the surface $\Gamma$. A few numerical results and comparisons with exact solutions are shown in Sections \ref{Section.Results} and \ref{Section.FFP}. Limitations, possible improvements and conclusions are discussed in Section \ref{Section.Conlusion}.

%%%%%%%%%%%%%%%%%%%%%%%%%%%%%%%%%%%%%%%%%%%%%
%%%%%%%%%%% NEW SECTION %%%%%%%%%%%%%%%%%%%%%
%%%%%%%%%%%%%%%%%%%%%%%%%%%%%%%%%%%%%%%%%%%%%
\section{Approximation of the outgoing Dirichlet--to--Neumann map} \label{Section.PDO}

We follow closely the approach described in \cite{Antoine1999}, but we only consider the pseudo--differential analysis with respect to the time variable $t$ and its Fourier dual $k$. We start with the wave equation,
\begin{align}
L u = \Delta u - c^{-2} \partial_{t}^{2} u = 0,
\end{align}
where the space derivatives can be written in terms of a local, tangential system of coordinates on the surface $\Gamma$ to obtain the symbol (with respect to time) of the wave operator,
\begin{align} \label{Eqn.LaplacianLocal}
\LL u = \partial_{r}^2u + 2 \HH  \partial_{r} u + \Delta_{\Gamma} u + k^2 u.
\end{align}
Here $\HH$ is the mean curvature of $\Gamma$, and $\partial_{r}$ represents the derivative in the outward normal direction on $\Gamma$. Multiplication by $k^2$ is an operator of order $+2$. Here $\Delta_{\Gamma}$ represents the Laplace--Beltrami operator on the surface $\Gamma$ which we consider to be of order $0$ with respect to $k$. A concise review of the differential geometry of surfaces, including the definition of curvatures and the Laplace--Beltrami operator, is found in Section 2 of \cite{Antoine1999}. We seek to decompose the wave field propagating across a surface into incoming and outgoing components. This can be done for convex surfaces. Based on Nirenberg's factorization theorem \cite{Nirenberg1973,Antoine1999}, there are two pseudo--differential operators $\Lambda^{-}$ and $\Lambda^{+}$ of order $+1$, such that  
\begin{align} \label{Eqn.Decomp01}
\LL = (\partial_{r} - \Lambda^{-})(\partial_{r} - \Lambda^{+}) u.
\end{align}
We refer to $\Lambda^{+}$ and $\Lambda^{-}$ as the outgoing and incoming Dirichlet--to--Neumann (DtN) operators for the radiating boundary value problem defined in the exterior of $\Gamma$. The operators $\Lambda^{\pm}$ have respective symbols $\lambda^{\pm}$ that admit the following expansion
\begin{align} \label{Eqn.SymbolExpansion01}
\lambda^{\pm} \sim \sum_{n = -1}^{+\infty} \lambda^{\pm}_{-n} 
\end{align}
where $\lambda^{\pm}_{-n}$ are symbols of degree $-n$ with respect to $k$. The expansion is understood in the sense of classical pseudo--differential operators. See \cite{Antoine1999,Chazarain1982-Book,Taylor2000-Book,Taylor1996-ChapterPSO} for details. We develop (\ref{Eqn.Decomp01}) to obtain
\begin{align} \label{Eqn.Decomp02}
\LL u = \partial_{r}^2 u - \left( \Lambda^{+} + \Lambda^{-} \right) \partial_{r} u  + \left( \Lambda^{-}\Lambda^{+} - \{ \partial_{r} \lambda^{+} \} \right) u.
\end{align}
By comparing (\ref{Eqn.LaplacianLocal}) and (\ref{Eqn.Decomp02}), and matching terms with same degree of radial derivatives, we find that,
\begin{align} 
& \Lambda^{+} + \Lambda^{-} = - 2 \HH , \label{Eqn.PseudoDiffSystem01} \\
& \Lambda^{-}\Lambda^{+} - \{ \partial_{r} \lambda^{+} \} = k^2 + \Delta_{\Gamma} . \label{Eqn.PseudoDiffSystem02}
\end{align}
Now we plug the expansion (\ref{Eqn.SymbolExpansion01}) to write the system (\ref{Eqn.PseudoDiffSystem01})--(\ref{Eqn.PseudoDiffSystem02}) in terms of pseudo--differential symbols. We obtain,
\begin{align} 
&\sum_{n=-1}^{+ \infty} \left( \lambda_{-n}^{+} + \lambda_{-n}^{-} \right) = - 2 \HH , \label{Eqn.SymbolSystem01} \\ 
& \sum_{n=-2}^{+ \infty} \sum_{j=-1}^{n+1}  \lambda^{-}_{-j} \lambda^{+}_{j-n}  - \sum_{m=-1}^{+ \infty}  \partial_{r} \lambda_{-m}^{+} = k^2 + \Delta_{\Gamma} . \label{Eqn.SymbolSystem02}
\end{align}

Now in (\ref{Eqn.SymbolSystem01}) we match symbols of the same order with respect to $k$ to find that
\begin{align} 
& \lambda_{-n}^{+} + \lambda_{-n}^{-} = 0 \qquad n \neq 0, \\
& \lambda_{0}^{+} + \lambda_{0}^{-} = - 2 \HH.
\end{align}
Similarly, matching the symbols of order $+2$ and $+1$ in (\ref{Eqn.SymbolSystem02}), we find that
\begin{align*} 
& \lambda_{1}^{-} \lambda_{1}^{+} = k^2, \\
& \lambda_{1}^{-}\lambda_{0}^{+} + \lambda_{0}^{-}\lambda_{1}^{+} - \partial_{r} \lambda_{1}^{+}  = 0.
\end{align*}
Since $\lambda_{1}^{-} = - \lambda_{1}^{+}$ and $\lambda_{0}^{-} = - 2 \HH - \lambda_{0}^{+}$ then we arrive at
\begin{align} 
& \lambda_{1}^{\pm} = \pm i k, \label{Eqn.Lp1}  \\
& \lambda_{0}^{\pm} = - \HH. \label{Eqn.L00}
\end{align}
Now we proceed to match the symbols of order $0$ from (\ref{Eqn.SymbolSystem02}) to obtain
\begin{align*} 
\lambda_{1}^{-} \lambda_{-1}^{+} +  \lambda_{0}^{-} \lambda_{0}^{+} + \lambda_{-1}^{-} \lambda_{1}^{+} - \partial_{r} \lambda_{0}^{+} = \Delta_{\Gamma}
\end{align*}
Using (\ref{Eqn.Lp1})--(\ref{Eqn.L00}) and the fact that $\partial_{r}\HH = \KK - 2 \HH^2$ (see  \cite[Probl. 11 \S 3.5]{DoCarmo1976}) we arrive at
\begin{align} 
\lambda_{-1}^{\pm} = \pm \frac{i}{2k} \left( \Delta_{\Gamma} - \HH^2 - \partial_{r} \HH \right) =  \pm \frac{i}{2k} \left( \Delta_{\Gamma} + \HH^2 - \KK \right) \label{Eqn.Lm1} 
\end{align}
where $\KK$ is the Gauss curvature of $\Gamma$. Using (\ref{Eqn.SymbolSystem02}) we leave the rest of the symbols expressed recursively as follows,
\begin{align} 
\lambda_{-n-1}^{+} = \frac{i}{2k} \left[ \partial_{r} \lambda_{-n}^{+} + \sum_{j=1}^{n-1} \lambda_{-j}^{+} \lambda_{j-n}^{+}  \right], \qquad n \geq 1.   \label{Eqn.LmRest} 
\end{align}

%%%%%%%%%%%%%%%%%%%%%%%%%%%%%%%%%%%%%%%%%%%%%
%%%%%%%%%%% NEW SECTION %%%%%%%%%%%%%%%%%%%%%
%%%%%%%%%%%%%%%%%%%%%%%%%%%%%%%%%%%%%%%%%%%%%
\section{Formal Generalization from Spherical Case}
\label{Section.SphericalCase}

The only impediment to algorithmically apply (\ref{Eqn.LmRest}) is the lack of knowledge of the dependence of $\lambda_{-n}^{+} = \lambda_{-n}^{+}(r)$ in order to evaluate $\partial_{r} \lambda_{-n}^{+}$. Here we seek to provide an explicit expression for the symbols $\lambda^{+}_{-n}$ as a function of $r$ for the case where $\Gamma$ is a sphere of radius $\ro$. 
We take a family of parallel spherical surfaces $\Gamma(r)$ parametrized by their radius $r$. Our practical application of the recursive formula (\ref{Eqn.LmRest}) is based on the following lemma.

\begin{Lemma} \label{Lemma.01}
For the family of spherical surfaces $\Gamma(r)$ of radius $r$, the symbols $\lambda_{-n}^{+}$ have the following dependence on the parameter $r$,
\begin{align} 
\lambda_{-n}^{+}(r) = \left(\frac{\ro}{r}\right)^{n+1} \lambda_{-n}^{+}(\ro) \qquad n=-1,0,1,2,...
\end{align}
where $\lambda_{-n}^{+}(\ro)$ is independent of $r$.
\end{Lemma}

\begin{proof}
The claim is clearly true for $n=-1$ even if $\Gamma$ is not a sphere. Now for the spherical case, the mean and Gauss curvatures have the following properties,
\begin{align*} 
\HH(r)  = \frac{\ro}{r} \HH(\ro) \qquad \text{and} \qquad \KK(r) = \left( \frac{\ro}{r} \right)^{2} \KK(\ro)
\end{align*}
where $\HH(\ro) = 1/\ro$ (mean curvature) and $\KK(\ro) = 1/\ro^2$ (Gauss curvature) for a sphere. Therefore, the claim is true for $n=0$. For a sphere, the Laplace--Beltrami operator satisfies $\Delta_{\Gamma(r)} = (\frac{\ro}{r})^2 \Delta_{\Gamma}$. Hence, the claim is also valid for $n=1$. Now we proceed by induction using the recursive formula (\ref{Eqn.LmRest}). Directly by inductive hypothesis, we have
\begin{align*} 
\lambda_{-(n+1)}^{+}(r) = \frac{i}{2k}  \left(\frac{\ro}{r}\right)^{n+2} \left[  - \frac{(n+1)}{\ro}  \lambda_{-n}^{+}(\ro) + \sum_{j=1}^{n-1} \lambda_{-j}^{+}(\ro) \lambda_{j-n}^{+}(\ro)  \right], \qquad n=1,2,...
\end{align*}

\end{proof}

Using Lemma \ref{Lemma.01}, we can express $\partial_{r} \lambda_{-n}^{+}$ explicitly in order to apply DtN operator $\lambda^{+}$ using the series (\ref{Eqn.SymbolExpansion01}). We obtain the following representation for the symbols on the surface $\Gamma$,
\begin{subequations} \label{Eqn.DtN0}
\begin{align} 
\lambda_{1}^{+} &= ik \label{Eqn.DtN0a} \\
\lambda_{0}^{+} &= - \HH \label{Eqn.DtN0b} \\
\lambda_{-1}^{+} &= \frac{i}{2k} \left( \Delta_{\Gamma} + \HH^2 - \KK  \right) \label{Eqn.DtN0c} \\
\lambda_{-(n+1)}^{+} &= \frac{i}{2k} \left[ -  (n+1) \HH \,  \lambda_{-n}^{+} + \sum_{j=1}^{n-1} \lambda_{-j}^{+} \lambda_{j-n}^{+}  \right], \qquad n=1,2,... \label{Eqn.DtN0d}
\end{align}
\end{subequations}
where we have identified $r = \HH^{-1}$. Although the recursive formula (\ref{Eqn.DtN0d}) is only valid for the spherical case, we formally generalize it to other geometries. In order to conform to the literature, we will use the following terminology in the sequel. Our approximation $[\lambda^{+}]_{N}$, of order $N$ with respect to $k$, for the symbol $\lambda^{+}$ of the outgoing DtN operator $\Lambda^{+}$ is given by
\begin{align} \label{Eqn.DtNorderN}
[\lambda^{+}]_{N} = \sum_{n=-1}^{N-1} \lambda_{-n}^{+} \qquad N \geq 0, 
\end{align}
where $\lambda_{-n}^{+}$ are defined by (\ref{Eqn.DtN0}). At least for the spherical case, we get $| \lambda^{+} - [\lambda^{+}]_{N} | \sim \mathcal{O}(k^{-N})$ as $k \to \infty$, which reveals the order of the approximation in a microlocal sense. Notice that $[\lambda^{+}]_{N}$ is not expected to converge as $N \to \infty$ for fixed $k$. However, as illustrated in Section \ref{Section.Results}, we do see more accurate results as $N$ increases within a reasonable range. We refer to $[\lambda^{+}]_{N}$ as the \textit{Dirichet--to--Neumann On--Surface Radiation Condition} of order $N$ (DtN--OSRC$_{N}$). Analogously, we refer to $[\lambda^{+}]_{N}^{-1}$ as the \textit{Neumann--to--Dirichlet On--Surface Radiation Condition} of order $N$ (NtD--OSRC$_{N}$). 

In order to easily compare with other results in the literature of OSRC, we display explicitly the DtN-OSRC$_{N}$ up to the fourth order, that is, for $N=0,1,2,3,4$:
\begin{subequations} \label{Eqn.DtNupto4}
\begin{align} 
[\lambda^{+}]_{0} &= ik \label{Eqn.DtN_0} \\
[\lambda^{+}]_{1} &= ik - \HH \label{Eqn.DtN_1} \\
[\lambda^{+}]_{2} &= ik - \HH + \frac{i}{2k} \left( \Delta_{\Gamma} + \HH^2 - \KK  \right) \label{Eqn.DtN_2} \\
[\lambda^{+}]_{3} &= ik - \HH + \left( \frac{i}{2k} + \frac{\HH}{2 k^2} \right) \left( \Delta_{\Gamma} + \HH^2 - \KK  \right)  \label{Eqn.DtN_3} \\
[\lambda^{+}]_{4} &= ik - \HH + \left( \frac{i}{2k} + \frac{\HH}{2 k^2} - \frac{i}{8 k^3} \left( \Delta_{\Gamma} + 7 \HH^2 - \KK  \right) \right) \left( \Delta_{\Gamma} + \HH^2 - \KK  \right)   \label{Eqn.DtN_4}
\end{align}
\end{subequations}

%%%%%%%%%%%%%%%%%%%%%%%%%%%%%%%%%%%%%%%%%%%%%
%%%%%%%%%%% NEW SECTION %%%%%%%%%%%%%%%%%%%%%
%%%%%%%%%%%%%%%%%%%%%%%%%%%%%%%%%%%%%%%%%%%%%
\section{Discrete Implementation} \label{Section.Numerics}

In this section, we describe a numerical implementation of the proposed on--surface radiation conditions defined by (\ref{Eqn.DtNorderN}). We numerically solve both the Dirichlet and the Neumann boundary value problems (\ref{Eqn.Main}), using a discrete version of the DtN--OSRC$_{N}$ and of NtD--OSRC$_{N}$, respectively. The approach is based on having a triangulation that approximates the surface $\Gamma$. In the following subsections, we describe the discrete approximations of the Laplace--Beltrami operator, of the mean and Gauss curvatures of $\Gamma$, and of the OSRC.

%%%%%%%%%%% NEW SECTION %%%%%%%%%%%%%%%%%%%%%
\subsection{Discrete Laplace--Beltrami operator and curvatures} \label{Subsection.LBCurv}

The subject of approximating geometrical properties and operators on triangulated surfaces has received considerable attention in the past decades. See  \cite{Petronetto2013,Reuter2009,Xu2006a,Xu2004,Xu2004a,Meyer2003,
Surazhsky2003,Borrelli2003,Magid2007,Bobenko2007,Rusinkiewicz2004,Grinspun2006,Sun2016-Book} for some recent advances in this area which we use as our guiding references.

The \emph{discrete Laplace--Beltrami operator} is described as follows. Let $\{\xx_{j}\}_{j=1}^{J}$ be the collection of vertices of the triangulation $\mathcal{T}$ of the surface $\Gamma$. For a fixed vertex $\xx_{j}$, let $\{\xx_{j(i)}\}_{i=1,2,...}$ be the neighboring vertices of $\xx_{j}$, and let
$\{e_{i}\}_{i=1,2,...}$ be the edges of the triangulation that connect $\xx_{j}$ and $\xx_{j(i)}$. Now for each edge $e_{i}$, let $\alpha_{i}$ and $\beta_{i}$ be angles at the vertex $\xx_{j}$ of the triangles on either side of the edge $e_{i}$. For a smooth function $u$ defined on the surface $\Gamma$, we use the following discrete Laplace--Beltrami operator,
\begin{align} \label{Eqn.DiscreteLB}
(\Delta_{\Gamma} u)(\xx_{j}) = \frac{2}{\pi^2} \sum_{i} \frac{\left( \sin \alpha_{i}  + \sin \beta_{i} \right)^2}{{\left( A_{j} A_{j(i)} \right)^{1/2}}} \left( u(\xx_{j(i)}) - u(\xx_{j}) \right),
\end{align}
where $A_{j}$ is the area associated with the vertex $\xx_{j}$ defined as $1/3$ of the total area of the triangular elements sharing the vertex $\xx_{j}$. 

For the \emph{discrete mean curvature}, we first define the discrete mean curvature vector as
\begin{align} \label{Eqn.DiscreteMCV}
\textbf{H} (\xx_{j}) = - \frac{1}{\pi^2} \sum_{i} \frac{\left( \sin \alpha_{i}  + \sin \beta_{i} \right)^2}{{\left( A_{j} A_{j(i)} \right)^{1/2}}} \left( \xx_{j(i)} - \xx_{j} \right).
\end{align}
Then the mean curvature is given by
\begin{align} \label{Eqn.DiscreteMC}
\HH (\xx_{j}) = \textbf{H} (\xx_{j}) \cdot \textbf{n}(\xx_{j})
\end{align}
where $\textbf{n}(\xx_{j})$ is the normal vector of the vertex $\xx_{j}$.

For the \emph{discrete Gauss curvature} at each vertex $\xx_{j}$ of the triangulation, we follow the approach described in \cite{Surazhsky2003,Magid2007,Xu2006a} based on the Gauss--Bonnet theorem. Let $\alpha_{i}$ be the angle between two successive edges sharing the vertex $\xx_{j}$. Then, the approximated Gauss curvature is given by
\begin{align} \label{Eqn.DiscreteGaussCurv}
\KK(\xx_{j}) = \frac{2 \pi - \sum_{i} \alpha_{i}}{A_{j}}
\end{align}
where $A_{j}$ is the area associated with the vertex $\xx_{j}$ as defined above.

%%%%%%%%%%% NEW SECTION %%%%%%%%%%%%%%%%%%%%%
\subsection{Discrete on--surface radiation conditions} \label{Subsection.DtN}

Given the triangulation $\mathcal{T}$ of the surface $\Gamma$, we have all the necessary definitions in Subsection \ref{Subsection.LBCurv} to implement a discrete version of the formulas (\ref{Eqn.DtN0}) to apply the discrete DtN--OSRC$_{N}$ given by (\ref{Eqn.DtNorderN}). On the triangulation $\mathcal{T}$, the operator $[\lambda^{+}]_{N}$ can be represented by a matrix of size $M \times M$, where $M$ is the number of vertices of the triangulation $\mathcal{T}$. This is equivalent to working on a piecewise linear function space supported by a mesh of triangular elements. We are not concerned with increasing the polynomial order of approximation because the error from the OSRC is known to dominate the calculations once the mesh is fine enough to resolve the wavelength $2 \pi / k$.

The recursive formula (\ref{Eqn.DtN0d}) and the fact that (\ref{Eqn.DtN0c}) contains the Laplace--Beltrami operator $\Delta_{\Gamma}$ imply that the symbols $\lambda_{-n}^{+}$ apply differential operators of increasingly higher order as $n$ increases. A direct application of the operators $\lambda_{-n}^{+}$ quickly becomes ill--conditioned (unstable) at the numerical level. Therefore, we propose a regularized application of the discrete DtN--OSRC$_{N}$ and NtD--OSRC$_{N}$ operators to obtain $u \approx \lambda^{+} f_{\rm Dir}$ or $v \approx (\lambda^{+})^{-1} f_{\rm Neu}$. This application is defined in precise term by the Algorithm \ref{Alg.DtN}.

\begin{algorithm}[H]
\begin{algorithmic}
\State {\tt Initialize}: $[\lambda^{+}] = \lambda_{1}^{+}$, {\tt and} $u = [\lambda^{+}] f_{\rm Dir}$ {\tt or} $v = [\lambda^{+}]^{-1} f_{\rm Neu}$.
\State {\tt for} $n = 0, 1, 2, \ldots,  N-1$
\State \hspace{2em} $ [\lambda^{+}] \leftarrow [\lambda^{+}] + \lambda_{-n}^{+}$
\State \hspace{2em} $u \leftarrow u + \RR_{n} \left( [\lambda^{+}]f_{\rm Dir} - u \right) \quad $ {\tt or} $\quad v \leftarrow v + \RR_{n} \left( [\lambda^{+}]^{-1}f_{\rm Neu} - v \right) $
\State {\tt end}
\State {\tt Return}: $u$ {\tt or} $v$.
\end{algorithmic}
\caption{: Regularized application of the $N^{\rm th}$ order approximate DtN--OSRC$_{N}$ operator $[\lambda^{+}]_{N}$ and NtD--OSRC$_{N}$ operator $[\lambda^{+}]_{N}^{-1}$. Here $\RR_{n}$ is a regularizer and $\lambda_{-n}^{+}$ is given by (\ref{Eqn.DtN0}).}
\label{Alg.DtN}
\end{algorithm}
In this paper, we use the following regularizer,
\begin{align} \label{Eqn.Regularizer}
\RR_{n} = \AA^{n}, \quad n \geq 0, \quad \text{where} \quad (\AA u)(\xx_{j}) = C_{j} \sum_{i} \frac{\left( \sin \alpha_{i}  + \sin \beta_{i} \right)^2}{{\left( A_{j} A_{j(i)} \right)^{1/2}}} u(\xx_{j(i)})
\end{align}
where $\xx_{j(i)}$, $\alpha_{i}$, $\beta_{j}$, $A_{i}$ and $A_{j(i)}$ are defined in Subsection \ref{Subsection.LBCurv}. The constant $C_{j}$ normalizes $\AA$ so that the entries in each row sum to $1$. In other words, $\AA$ is a local averaging operator mimicking the off--diagonal entries of the Laplace--Beltrami operator (\ref{Eqn.DiscreteLB}). Other approaches, such as global smoothing operators, may work too. For instance, good results are obtained by using,
\begin{align} \label{Eqn.Regularizer2}
\RR_{n} = \left( I - c_{n} \Delta_{\Gamma} \right)^{-1}
\end{align}
where $c_{n} \sim \left(n/k \right)^2$. In practice, the application of (\ref{Eqn.Regularizer2}) requires the inversion of an $M \times M$ matrix representing the discretization of a coercive operator on $\Gamma$. In this paper we show results using the application of (\ref{Eqn.Regularizer}), which is computationally lighter than (\ref{Eqn.Regularizer2}). The regularizers (\ref{Eqn.Regularizer}) and (\ref{Eqn.Regularizer2}) damp out the spurious oscillations that may appear due to non--smooth imperfections in approximating the surface $\Gamma$, the curvatures $\HH$ and $\KK$, as well as the operator $\Delta_{\Gamma}$. It turns out that approximation of geometrical properties on triangulations is slowly convergent and quite sensitive to the valence of vertices and to the quality of triangles \cite{Xu2006a,Xu2004a,Xu2004,Borrelli2003,Magid2007,Bobenko2007}.

%%%%%%%%%%%%%%%%%%%%%%%%%%%%%%%%%%%%%%%%%%%%%
%%%%%%%%%%% NEW SECTION %%%%%%%%%%%%%%%%%%%%%
%%%%%%%%%%%%%%%%%%%%%%%%%%%%%%%%%%%%%%%%%%%%%
\section{Numerical Results} \label{Section.Results}

In this section we present a few numerical results obtained from the implementation of the proposed high order OSRC described in Section \ref{Section.Numerics}. We discuss three realizations for the surface $\Gamma$: the unit sphere, a marshmallow--like surface and a couple of more challenging non--convex surfaces in the likeness of a squash and a red blood cell. These four shapes are shown in Figure \ref{Fig:Meshes}. The approximate mean curvature, defined by (\ref{Eqn.DiscreteMC}), is shown in Figure \ref{Fig:MeanCurv}. Although the pseudo--differential approximation of the DtN operator described in Section \ref{Section.PDO} is valid for convex surfaces, we still apply it to non--convex surfaces to explore its performance.

\begin{figure}[H]
\centering
\subfloat[Sphere]{\includegraphics[height=0.35\textwidth]{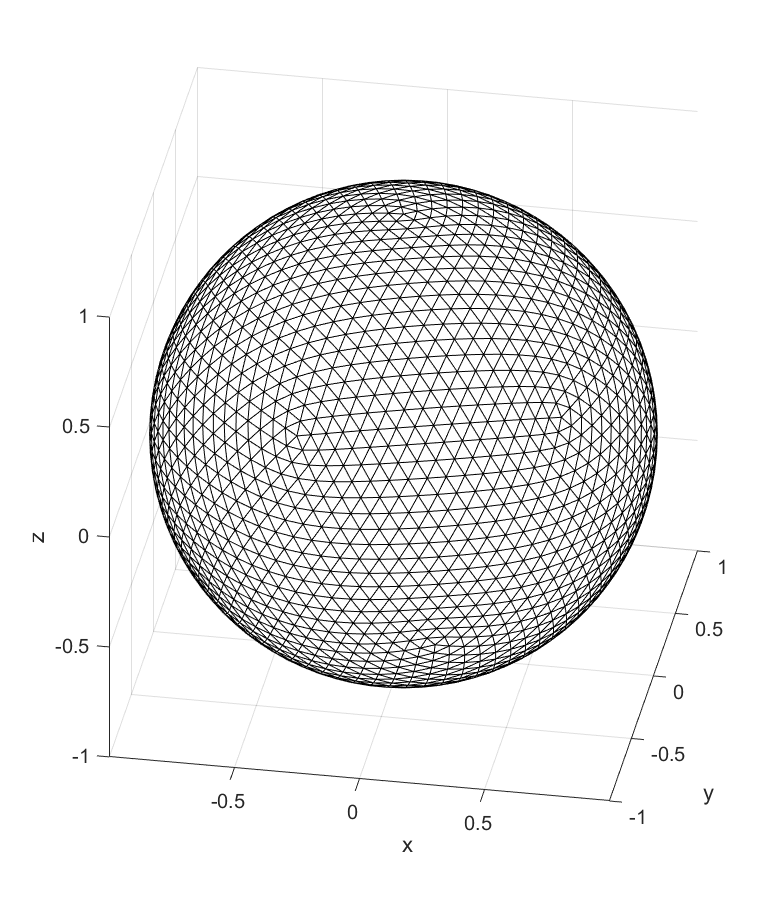}}
\subfloat[Marshmallow]{\includegraphics[height=0.35\textwidth]{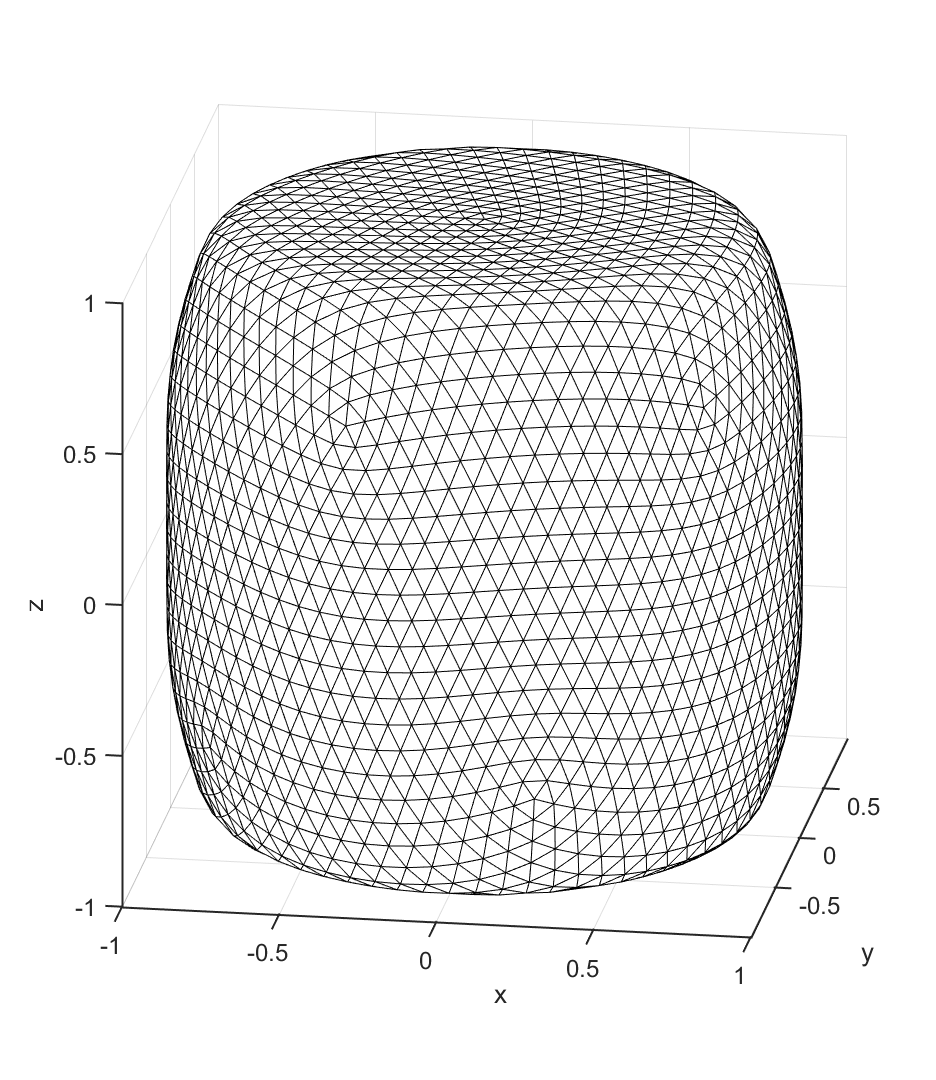}} \\
\subfloat[Squash]{\includegraphics[height=0.35\textwidth]{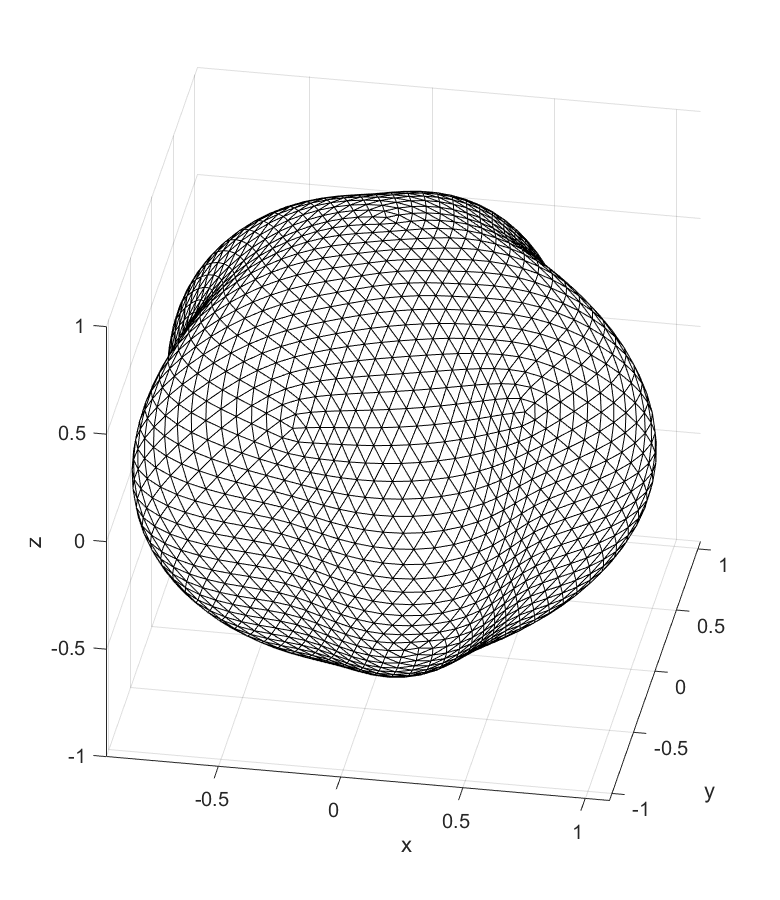} }
\subfloat[Blood cell]{\includegraphics[height=0.35\textwidth]{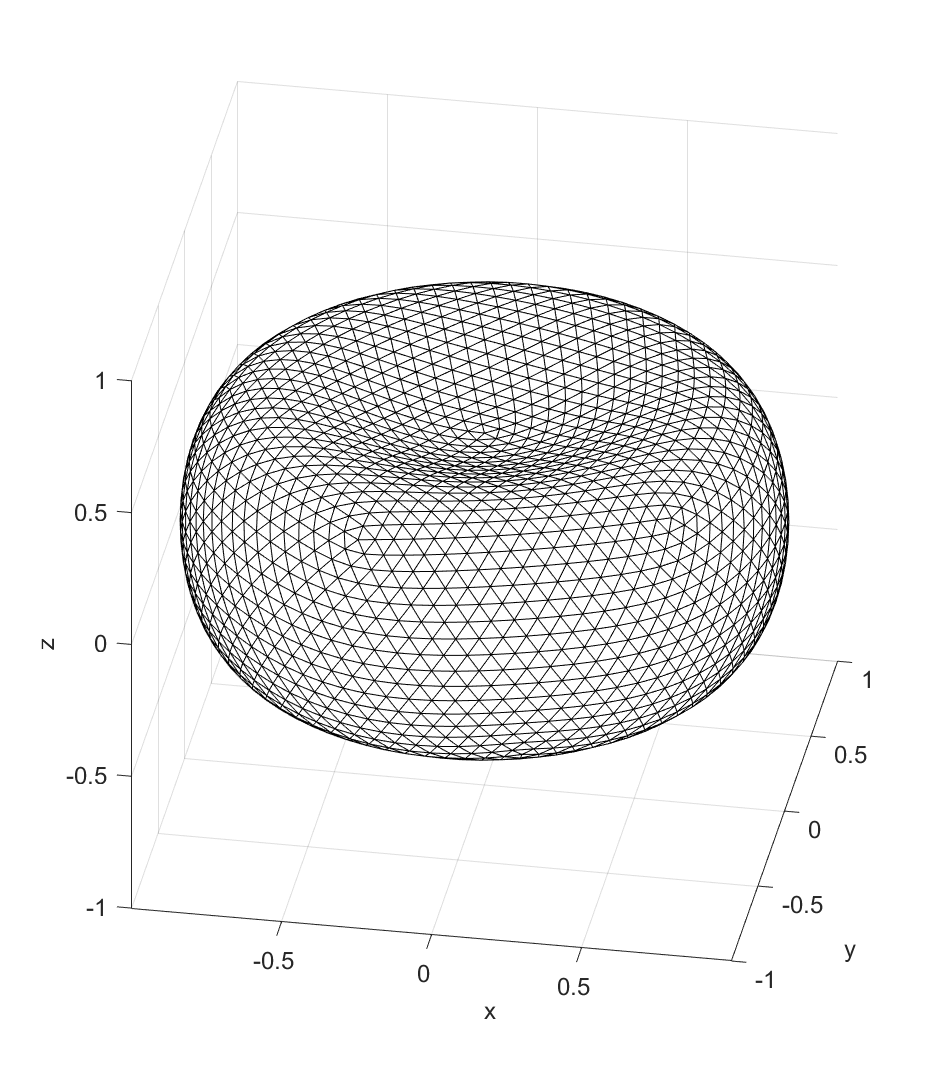} }
\caption{Coarse meshes for a sphere, a marshmallow, a squash and a blood cell. These meshes contain 2562 vertices and 5120 triangles.}  
\label{Fig:Meshes}
\end{figure}

\begin{figure}[H]
\centering
\subfloat[Marshmallow]{\includegraphics[height=0.25\textwidth]{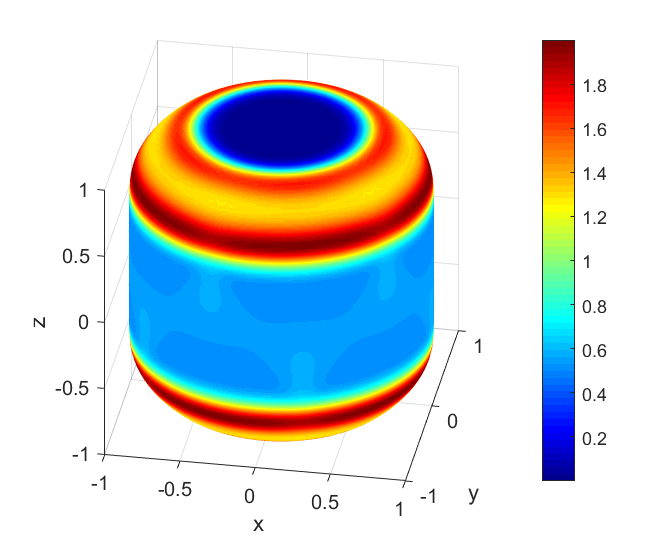}}
\subfloat[Squash]{\includegraphics[height=0.25\textwidth]{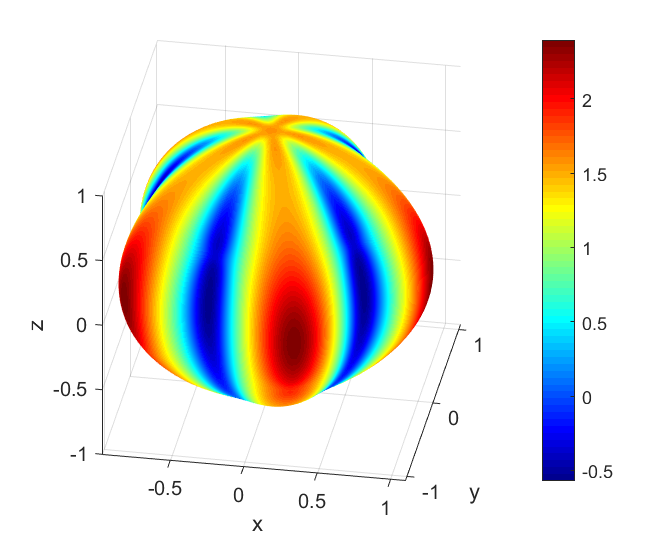}}
\subfloat[Blood cell]{\includegraphics[height=0.25\textwidth]{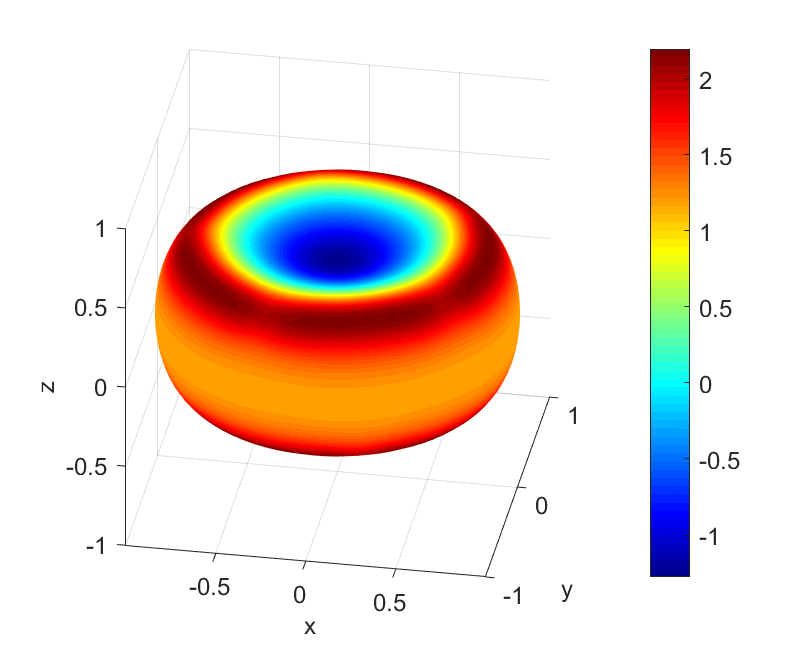}}
\caption{Mean curvature on the marshmallow, the squash and the blood cell. Notice that the squash and blood cell are not convex surfaces, and they have regions of negative mean curvature.}  
\label{Fig:MeanCurv}
\end{figure}

For comparison, we want to have an exact solution to the boundary values problems (\ref{Eqn.Main}), regardless of the geometry of the surface $\Gamma$. Hence, we define very specific boundary data in order to engineer an exact solution. First consider the field,
\begin{align} \label{Eqn.InternalSource}
F(\xx) = \sum_{l=1}^{L} \Phi(\xx - \cc_{l})
\end{align}
where $\Phi(\xx) = e^{ik|\xx|}/(4\pi|\xx|)$ is the outgoing fundamental solution to the Helmholtz equation. Hence, $F$ represents the superposition of $L$ point--sources with respective locations at $\cc_{l}$. If the points $\cc_{l}$ for $l=1,2,...,L$ are located inside the surface $\Gamma$, then the exact solution to the boundary values problems (\ref{Eqn.Main}) with Dirichlet data $f_{\rm Dir} = F|_{\Gamma}$ or Neumann data $f_{\rm Neu} = \partial_{\nu} F$ is given by $F|_{\Omega}$. Figure \ref{Fig:Error} displays the $L^2$--norm relative errors for both DtN--OSRC$_{N}$ and NtD--OSRC$_{N}$ as $N$ increases, for $k = 10 \pi$, and using (\ref{Eqn.Regularizer}) as the regularizer. In order to create an interesting (non--trivial) problem, we used $L=5$ point sources to define (\ref{Eqn.InternalSource}). These 5 point sources are located on the corners of a regular pentagon of circumradius equal to $1/2$ and lying on the $\{z=0\}$ plane.  The errors are defined as follows,
\begin{align} \label{Eqn.Error}
\epsilon_{\rm DtN} = \frac{\| f_{\rm Neu} - [\lambda^{+}]_{N} f_{\rm Dir} \|_{L^{2}(\Gamma)}}{\| f_{\rm Neu} \|_{L^{2}(\Gamma)}} \quad \text{and} \quad \epsilon_{\rm NtD} = \frac{\| f_{\rm Dir} - [\lambda^{+}]_{N}^{-1} f_{\rm Neu} \|_{L^{2}(\Gamma)}}{\| f_{\rm Dir} \|_{L^{2}(\Gamma)}}.
\end{align}
We observe in Figure \ref{Fig:Error} that as the order $N$ of the OSRC$_{N}$ increases, the error decreases. However, at some point the error becomes stagnant. This is not the result of the finite--dimensional approximation obtained by the discrete triangulation. A refinement of the mesh does not diminish the error any further. Nevertheless, we are pleased to see that the higher order OSRC diminishes the error in a systematic manner. In the spherical case, we are able to approximate the solution with an error below $1\%$. For the other three surfaces, the error drops below $2\%$, $4\%$ and $10\%$, respectively.  

\begin{figure}[H]
\centering
\subfloat[Sphere]{\includegraphics[width=0.35\textwidth]{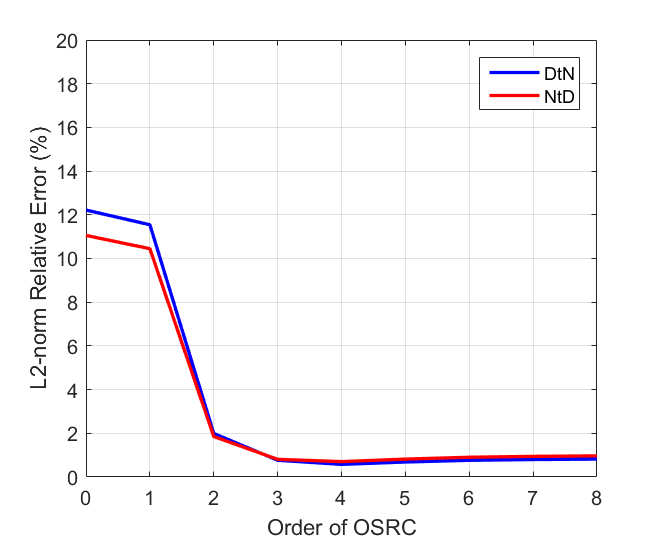}}
\subfloat[Marshmallow]{\includegraphics[width=0.35\textwidth]{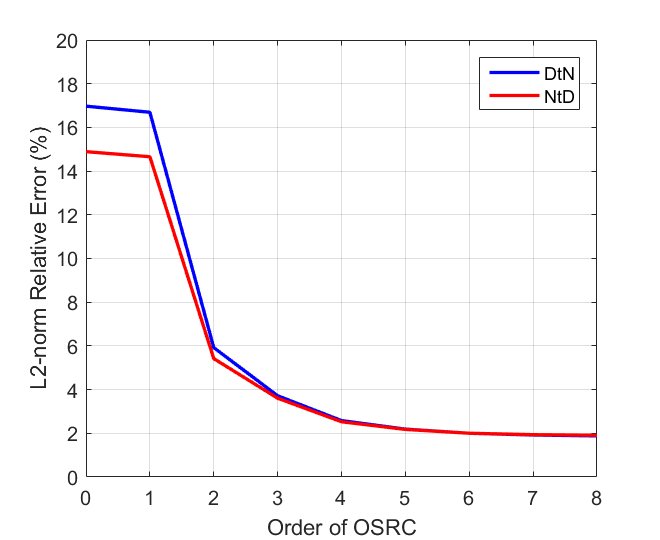}} \\
\subfloat[Squash]{\includegraphics[width=0.35\textwidth]{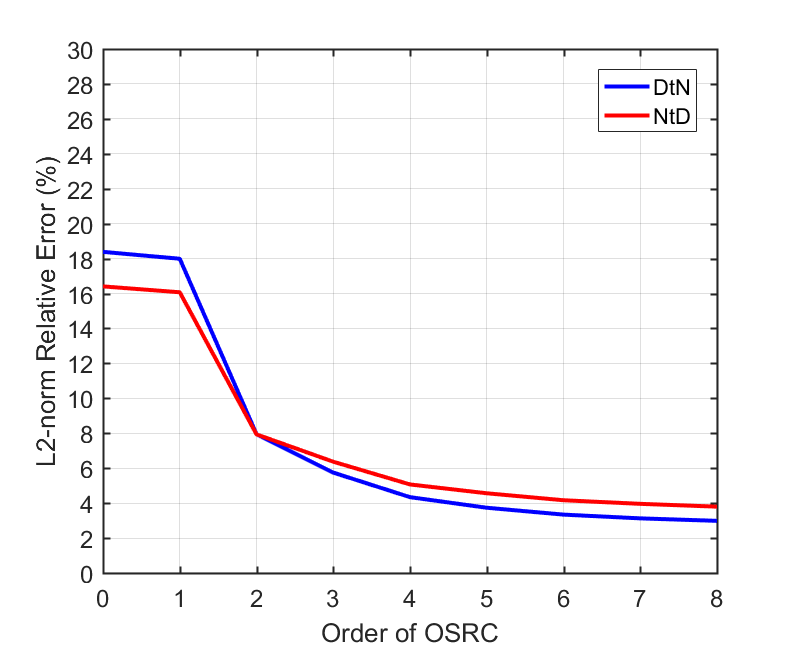}}
\subfloat[Blood cell]{\includegraphics[width=0.35\textwidth]{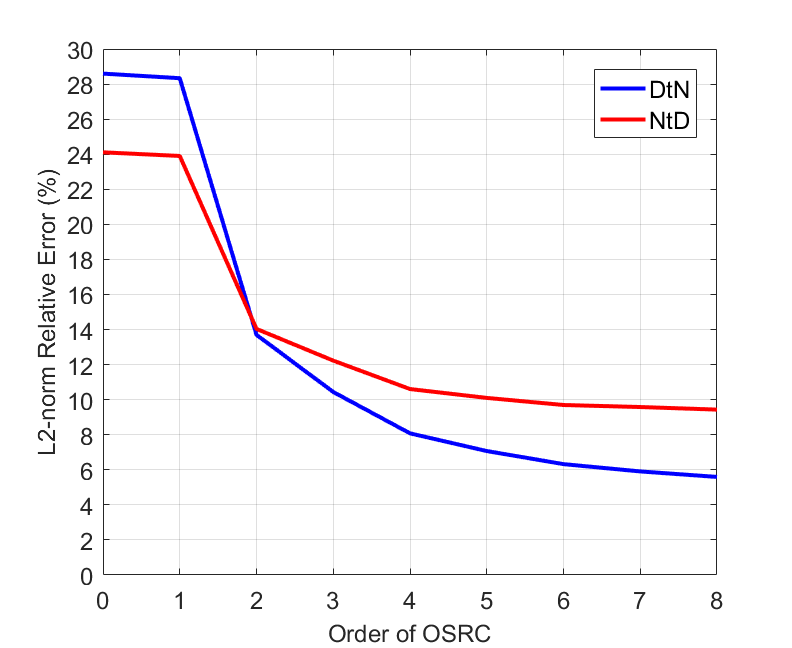}}
\caption{Relative error (\ref{Eqn.Error}) for solving the Dirichlet and Neumann boundary value problems on the sphere, squash, marshmallow, and blood cell surfaces. The wave number is $k = 10 \pi$ and the regularizer is (\ref{Eqn.Regularizer}). The order of the OSRC is defined in Section \ref{Section.SphericalCase}. The number of vertices in each mesh is $M = 10242$ which renders approximately $12$ points per wavelength.}  
\label{Fig:Error}
\end{figure}

As defined by the Algorithm \ref{Alg.DtN}, the application of the NtD--OSRC requires the inversion of the operators $[\lambda_{+}]_{N}$. However, this inversion is not a heavy computational burden because this auxiliary problem is two--dimensional and local in nature. So its discretization leads to medium--size sparse matrices. In our experience, the memory demand is low and Gaussian elimination works very well. Better efficiency can be achieved by using permutation schemes to pre--order the matrices for Gaussian elimination. The sparsity patterns of the matrices representing the discretization of $[\lambda^{+}]_{N}$ for $N = 2,4,6$ are displayed in Figure \ref{Fig:Sparsity}. The top three panels display the raw patterns induced by the connectivity of the triangular mesh. Due to the nature of the recursive formula (\ref{Eqn.DtN0}), the pattern of $[\lambda^{+}]_{2N+1}$ is the same as the pattern of $[\lambda^{+}]_{2N}$. The appearance of the Laplace--Beltrami operator $\Delta_{\Gamma}$ in (\ref{Eqn.DtN0c}) causes the matrix to lose sparsity as $N$ increases. Therefore, we computed the reverse Cuthill--McKee permutation (built in MATLAB) of the matrix representing the Laplace--Beltrami operator $\Delta_{\Gamma}$, and applied the same permutation to $[\lambda^{+}]_{N}$. The sparsity patterns for the re--ordered matrices are shown in the bottom panels of Figure \ref{Fig:Sparsity}. Notice the considerable improvement in reducing the bandwidth of these matrices. Significant improvements in computational efficiency can be expected when using Gaussian elimination to apply the NtD--OSRC$_{N}$ for fine meshes. 

\begin{figure}[H]
\centering
\subfloat[N=2]{\includegraphics[width=0.32\textwidth]{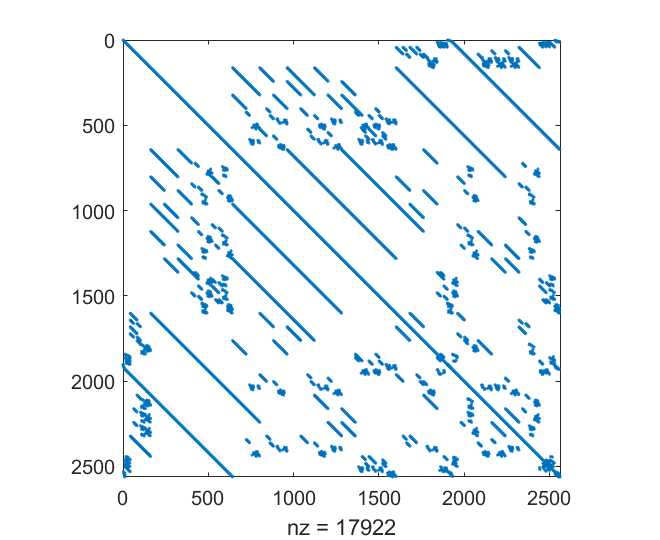}}
\subfloat[N=4]{\includegraphics[width=0.32\textwidth]{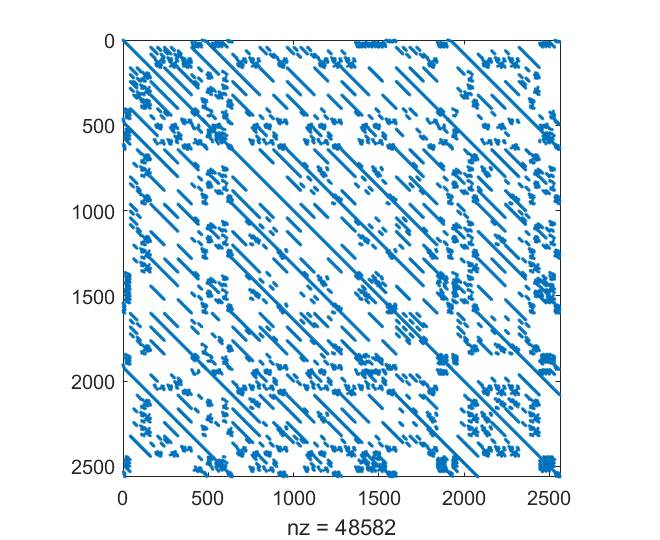}}
\subfloat[N=6]{\includegraphics[width=0.32\textwidth]{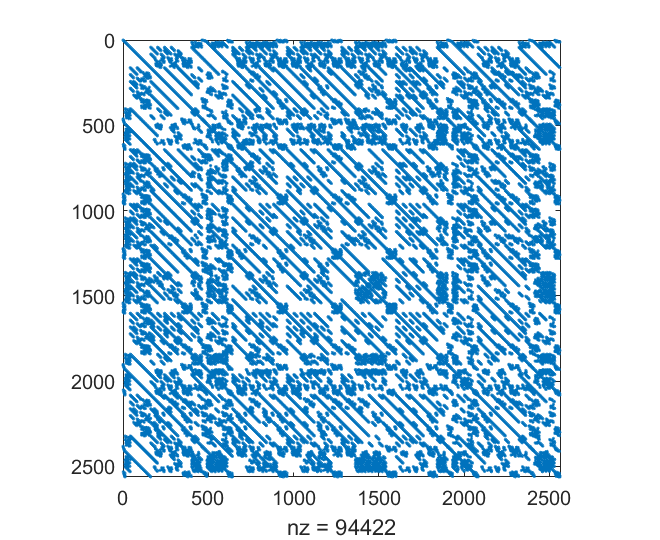}} \\
\subfloat[N=2]{\includegraphics[width=0.32\textwidth]{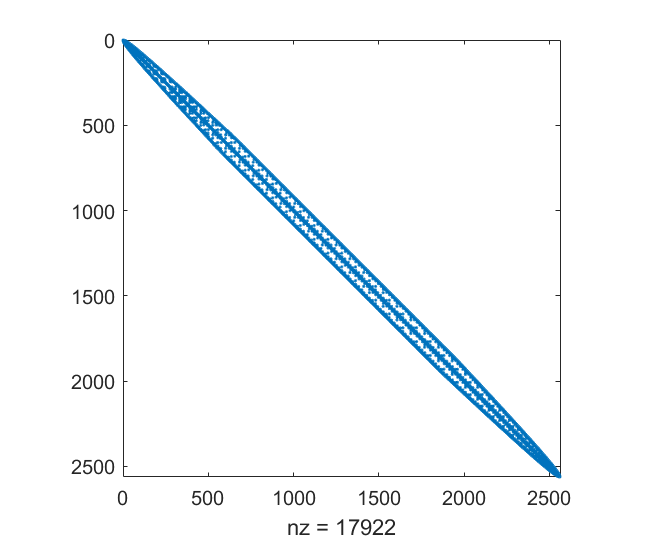}}
\subfloat[N=4]{\includegraphics[width=0.32\textwidth]{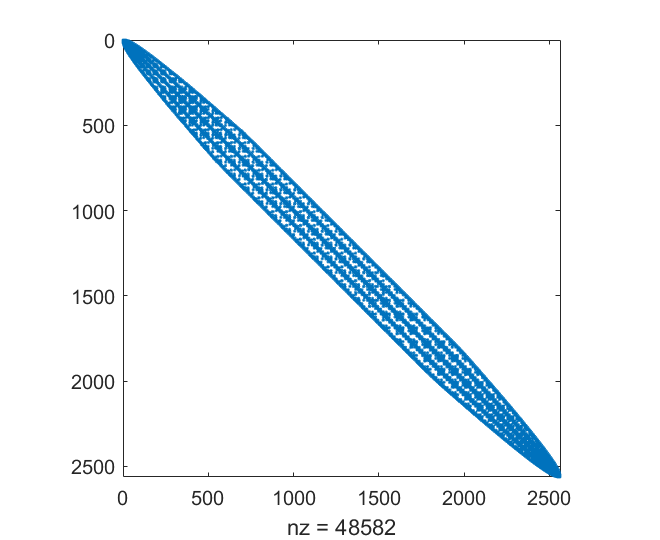}}
\subfloat[N=6]{\includegraphics[width=0.32\textwidth]{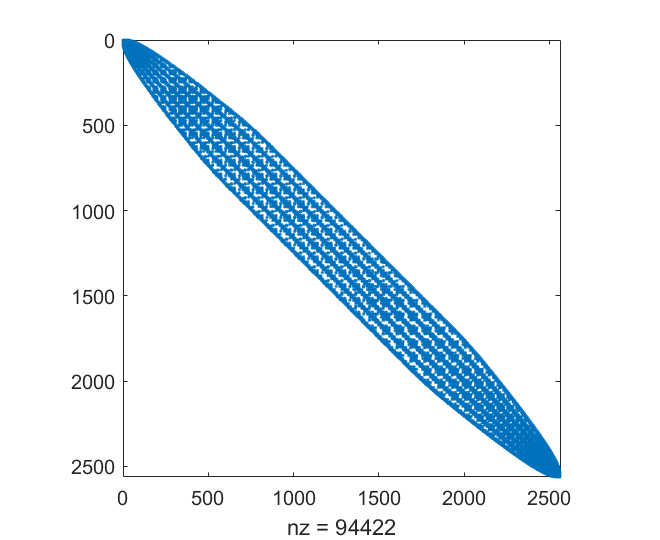}}
\caption{Sparsity patterns for the discrete version of $[\lambda^{+}]_{N}$ for $N = 2,4,6$, and the corresponding results using the reversed Cuthill--McKee permutation for the Laplace--Beltrami operator $\Delta_{\Gamma}$.}  
\label{Fig:Sparsity}
\end{figure}

%%%%%%%%%%%%%%%%%%%%%%%%%%%%%%%%%%%%%%%%%%%%%
%%%%%%%%%%% NEW SECTION %%%%%%%%%%%%%%%%%%%%%
%%%%%%%%%%%%%%%%%%%%%%%%%%%%%%%%%%%%%%%%%%%%%
\section{Far--field Pattern} \label{Section.FFP}
Now we illustrate the ability of the proposed high order OSRC to approximate the far--field pattern radiated from the surface $\Gamma$. The far--field pattern $u_{\infty}$ of the solution to the boundary value problem (\ref{Eqn.Main}) can be obtained from the Green's formula (\ref{Eqn.Green}) and the following asymptotic behavior of the fundamental solution
\begin{subequations} \label{Eqn.FFG}
\begin{align} 
\Phi(\xx,\yy) &= \frac{e^{ik|\xx|}}{4 \pi |\xx|} \left( e^{-ik \hat{\xx} \cdot \yy} + \mathcal{O}(|\xx|^{-1}) \right), \\
\frac{\partial \Phi(\xx,\yy)}{\partial \nu(\yy)} &= \frac{e^{ik|\xx|}}{4 \pi |\xx|} \left(  -ik \hat{\xx} \cdot \nu(\yy)   e^{-ik \hat{\xx} \cdot \yy} + \mathcal{O}(|\xx|^{-1})  \right).
\end{align}
\end{subequations}
We get
\begin{align} \label{Eqn.FFP}
u_{\infty}(\hat{\xx}) = - \int_{\Gamma} \left( ik \hat{\xx} \cdot \nu(\yy) u(\yy)  + \partial_{\nu} u(\yy)  \right) e^{-ik \hat{\xx} \cdot \yy} dS(\yy), \qquad |\hat{\xx}| = 1.
\end{align}

For the Dirichlet problem, $u$ is prescribed and $\partial_{\nu} u$ is approximated using the DtN--OSRC$_{N}$ operator. For the Neumann problem, $\partial_{\nu} u$ is prescribed and $u$ approximated using the NtD--OSRC$_{N}$ operator. Then the computation of the far--field pattern reduces to numerical quadrature to approximate the integral in (\ref{Eqn.FFP}). The behavior of the error in approximating the far--field pattern with the proposed method is displayed in Figure \ref{Fig:ErrorFFP} for increasing order $N$. This relative error is measured in the $L^2$--norm on the unit sphere (domain of the far--field pattern) for both the DtN--OSRC$_{N}$ and NtD--OSRC$_{N}$ methods for the same surfaces and parameters discussed in Section \ref{Section.Results} and Figure \ref{Fig:Error}.

\begin{figure}[H]
\centering
\subfloat[Sphere]{\includegraphics[width=0.35\textwidth]{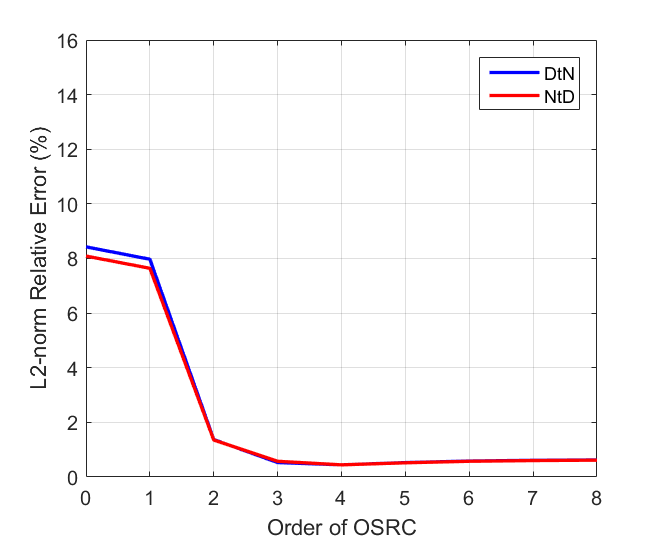}}
\subfloat[Marshmallow]{\includegraphics[width=0.35\textwidth]{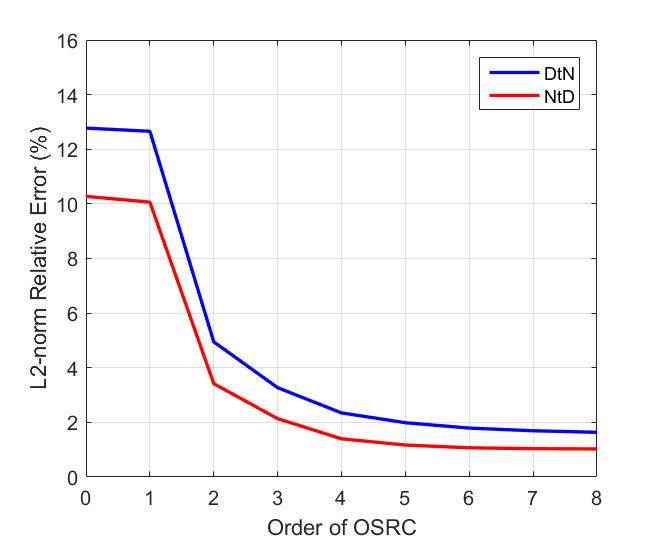}} \\
\subfloat[Squash]{\includegraphics[width=0.35\textwidth]{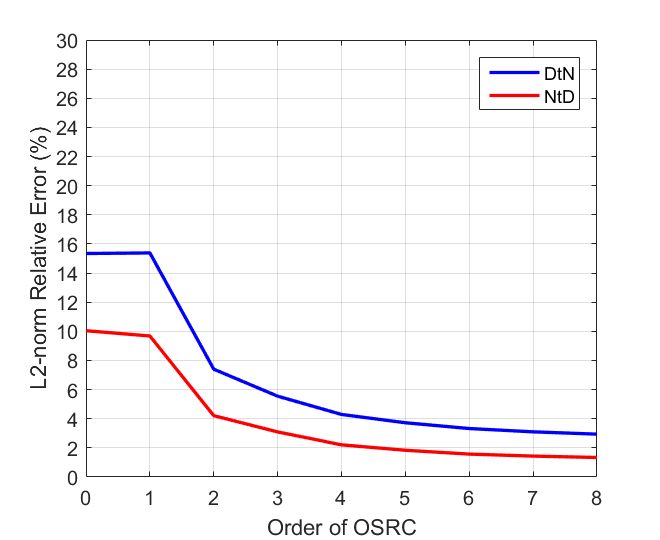}}
\subfloat[Blood cell]{\includegraphics[width=0.35\textwidth]{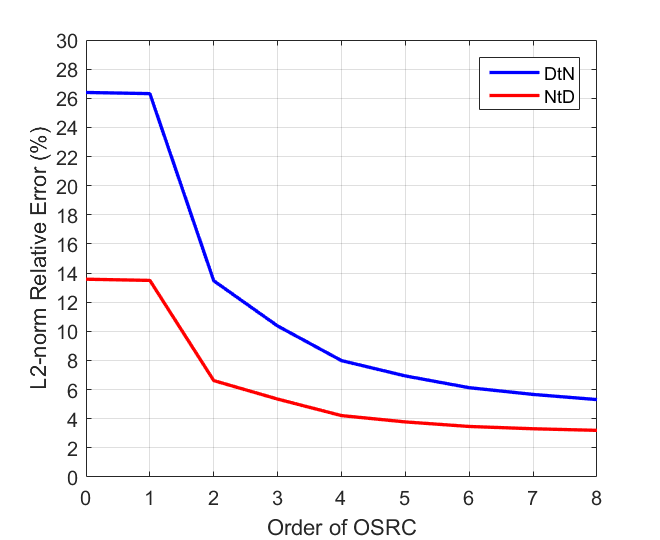}}
\caption{Relative error in approximating the far--field pattern (\ref{Eqn.FFP}) for solving the Dirichlet and Neumann boundary value problems on the sphere, squash, marshmallow, and blood cell surfaces. The wave number is $k = 10 \pi$ and the regularizer is (\ref{Eqn.Regularizer}). The order of the OSRC is defined in Section \ref{Section.SphericalCase}. The number of vertices in each mesh is $M = 10242$ which renders approximately $12$ points per wavelength.}  
\label{Fig:ErrorFFP}
\end{figure}

Figures \ref{Fig:FFP10pi} and \ref{Fig:FFP20pi} display the far--field patterns for $k=10\pi$ and $k=20\pi$, respectively, for the squash--like surface. A comparison is shown between the exact and approximate solution using the DtN--OSRC$_{8}$. The Dirichlet data is prescribed as in Section \ref{Section.Results}. The presence of the sources inside $\Gamma$ produces an interesting interference pattern that the OSRC is able to capture quite accurately. 

\begin{figure}[H]
\centering
\includegraphics[width=0.95\textwidth, trim={50 10 50 10}]{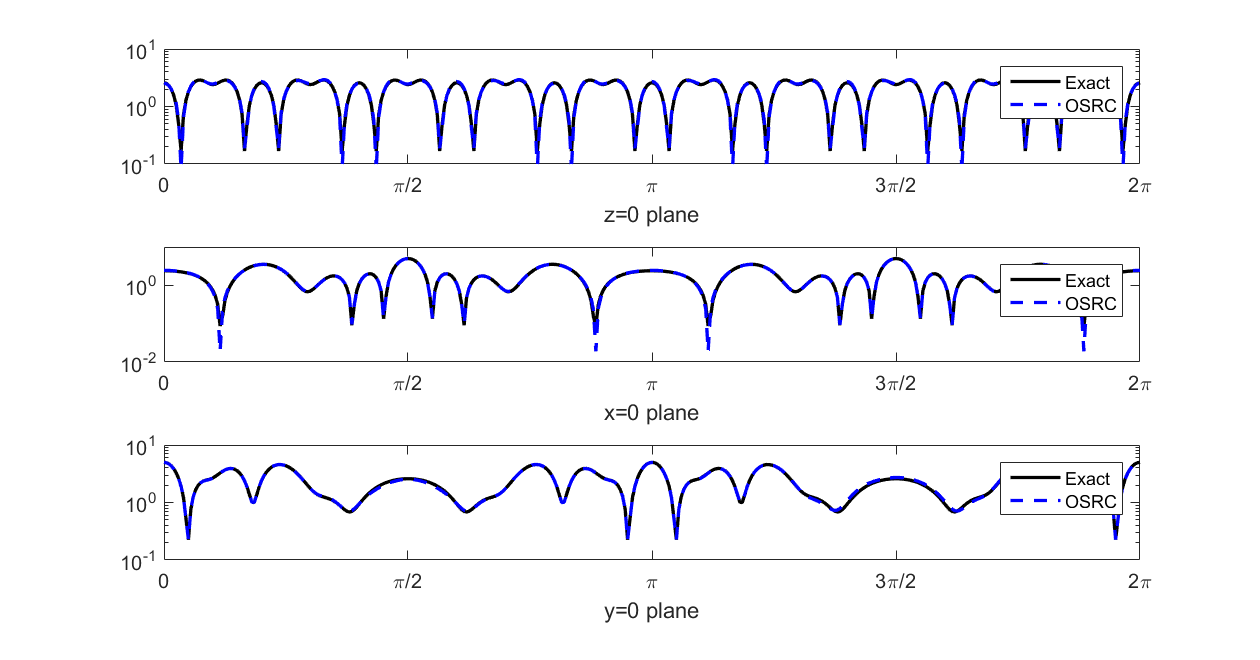}
\caption{Absolute value of the far--field pattern (\ref{Eqn.FFP}) for the Dirichlet problem on the squash--like surface. The exact boundary data is $u = F|_{\Gamma}$ and $\partial_{\nu} u = \partial_{\nu} F$, and $F$ is given as in Section \ref{Section.Results}. For the DtN--OSRC approximation, in (\ref{Eqn.FFP}) we replace $\partial_{\nu} u$ with $[\lambda^{+}]_{8} f_{\rm Dir}$ which is computed according to Algorithm \ref{Alg.DtN}. The wavenumber $k=10 \pi$. The $L^2$--norm relative error for the far--field pattern is $2.94 \%$.}  
\label{Fig:FFP10pi}
\end{figure}

\begin{figure}[H]
\centering
\includegraphics[width=0.95\textwidth, trim={50 10 50 10}]{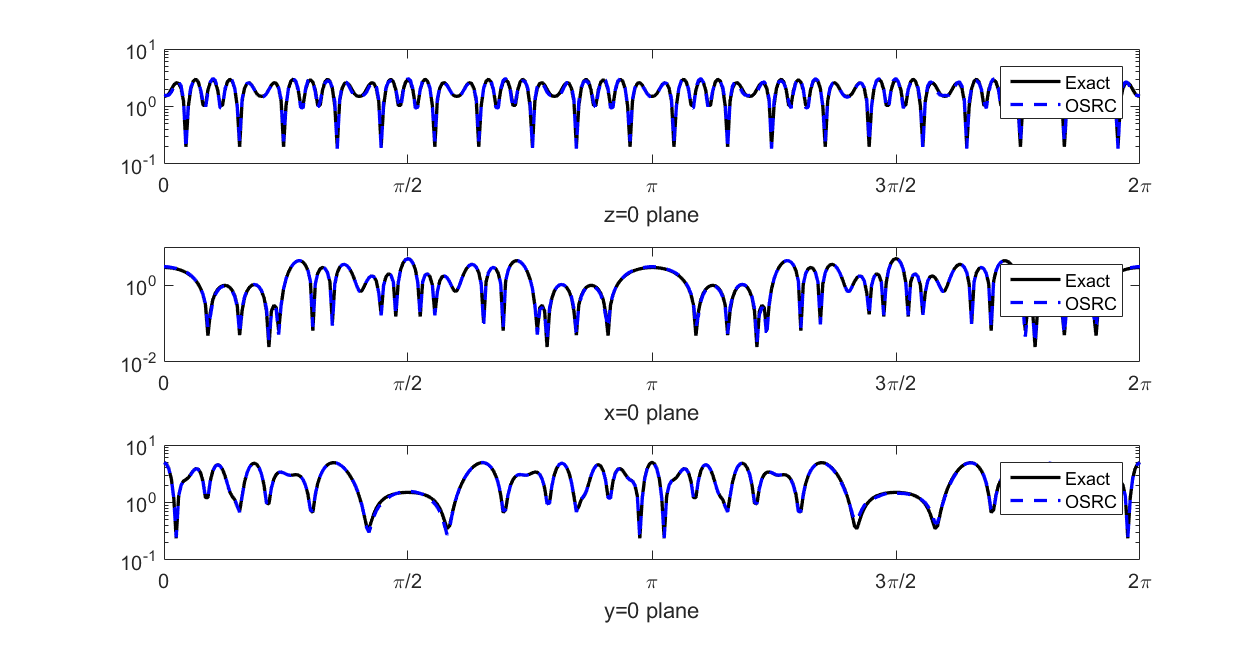}
\caption{Absolute value of the far--field pattern (\ref{Eqn.FFP}) for the Dirichlet problem on the squash--like surface. The exact boundary data is $u = F|_{\Gamma}$ and $\partial_{\nu} u = \partial_{\nu} F$, and $F$ is given as in Section \ref{Section.Results}. For the DtN--OSRC approximation, in (\ref{Eqn.FFP}) we replace $\partial_{\nu} u$ with $[\lambda^{+}]_{8} f_{\rm Dir}$ which is computed according to Algorithm \ref{Alg.DtN}. The wavenumber $k=20 \pi$. The $L^2$--norm relative error for the far--field pattern is $2.35 \%$.}  
\label{Fig:FFP20pi}
\end{figure}

%%%%%%%%%%%%%%%%%%%%%%%%%%%%%%%%%%%%%%%%%%%%%
%%%%%%%%%%% NEW SECTION %%%%%%%%%%%%%%%%%%%%%
%%%%%%%%%%%%%%%%%%%%%%%%%%%%%%%%%%%%%%%%%%%%%
\section{Conclusions} \label{Section.Conlusion}

We have constructed an on--surface radiation condition (OSRC$_N$) of arbitrarily high order $N$. This condition is applicable on smooth convex surfaces of arbitrary shape. The high order is obtained through a recursive formula for higher order pseudo--differential symbols of the outgoing Dirichlet--to--Neumann (DtN) map. From the numerical experiments, we observe a systematic reduction of the error as the order $N$ increases which is the main achievement of this paper. This is a rather impressive property of the DtN--OSRC$_{N}$, considering that we are applying local, two--dimensional, explicit operations to solve a three--dimensional elliptic boundary value problem.

The major limitation of our proposed high order OSRC is caused by the formal generalization  of the recursive formula (\ref{Eqn.DtN0d}) from the spherical case to other geometries. See Section \ref{Section.SphericalCase} for details. The symbols (\ref{Eqn.DtN0a})--(\ref{Eqn.DtN0c}) are valid for arbitrary convex geometry. But the recursion (\ref{Eqn.DtN0d}) is not. However, we suspect that some of the correct geometric information from (\ref{Eqn.DtN0c}) gets propagated into the higher symbols through (\ref{Eqn.DtN0d}). 
The other limitations are similar to those of previous formulations of OSRC. In general, the OSRC provides a relatively crude approximation of the solution lacking a--priori error estimates. The microlocal analysis employed to define the geometric OSRC does not imply that the DtN--OSRC$_{N}$ converges to the exact DtN operator as $N$ increases for fixed $k$. The only notion of convergence is obtained as $k \to \infty$. And even in that case, the convergence in the microlocal sense allows for a mismatch in the form of a smoothing operator. This may explain the stagnation of the error observed in Figure \ref{Fig:Error}.

Finally, we propose several directions of research that may improve or extend our work:
(i) A generalization of Lemma \ref{Lemma.01} to ellipsoidal geometry may provide more accurate results including the explicit appearance of the principal curvatures in the recursive formula. (ii) We expect the general ideas of this paper to apply to electromagnetic waves; an area where many engineering application reside. (iii) In this paper, we are not particularly concerned with the error associated with the triangulation of the surface $\Gamma$ and the discrete approximation of geometric properties. Variational and finite element formulations compatible with geometric structures may offer great improvements \cite{Arnold2006}. (iv) Following the approaches from \cite{Acosta2015f,Alzubaidi2016}, problems with multiple disjoint obstacles may also be addressed.

\section*{References}

%% `Elsevier LaTeX' style
\bibliographystyle{elsarticle-num}
\bibliography{Waves}

\begin{thebibliography}{10}
\expandafter\ifx\csname url\endcsname\relax
  \def\url#1{\texttt{#1}}\fi
\expandafter\ifx\csname urlprefix\endcsname\relax\def\urlprefix{URL }\fi
\expandafter\ifx\csname href\endcsname\relax
  \def\href#1#2{#2} \def\path#1{#1}\fi

\bibitem{Kress-Book-1999}
R.~Kress, {Linear Integral Equations}, 2nd Edition, Vol.~82 of Applied
  mathematical sciences, New York : Springer-Verlag, 1999.

\bibitem{McLean2000}
W.~McLean, {Strongly Elliptic Systems and Boundary Integral Equations},
  Cambridge Univ. Press, 2000.

\bibitem{Martin-Book-2006}
P.~Martin, {Multiple Scattering}, Encyclopedia of Mathematics and its
  Application vol. 107, Cambridge Univ. Press, 2006.

\bibitem{Chen2010-Book}
G.~Chen, J.~Zhou, {Boundary element methods with applications to nonlinear
  problems}, Springer Science {\&} Business Media, 2010.

\bibitem{Rabinovich2010}
D.~Rabinovich, D.~Givoli, E.~Becache, {Comparison of high-order absorbing
  boundary conditions and perfectly matched layers in the frequency domain},
  Int. j. numer. method. biomed. eng. 26~(1) (2010) 1351----1369.
\newblock \href {http://dx.doi.org/10.1002/cnm.1394}
  {\path{doi:10.1002/cnm.1394}}.

\bibitem{Givoli2004}
D.~Givoli, {High-order local non-reflecting boundary conditions: A review},
  Wave Motion 39~(4) (2004) 319--326.
\newblock \href {http://dx.doi.org/10.1016/j.wavemoti.2003.12.004}
  {\path{doi:10.1016/j.wavemoti.2003.12.004}}.

\bibitem{Tsynkov1998}
S.~Tsynkov, {Numerical solution of problems on unbounced domains. A review},
  Appl. Numer. Math. 27 (1998) 465--532.
\newblock \href {http://dx.doi.org/10.1016/S0168-9274(98)00025-7}
  {\path{doi:10.1016/S0168-9274(98)00025-7}}.

\bibitem{Ihlenburg1998-Book}
F.~Ihlenburg, {Finite element analysis of acoustic scattering}, Springer-Verlag
  : New York, 1998.

\bibitem{Zarmi2013}
A.~Zarmi, E.~Turkel, {A general approach for high order absorbing boundary
  conditions for the Helmholtz equation}, J. Comput. Phys. 242 (2013) 387--404.
\newblock \href {http://dx.doi.org/10.1016/j.jcp.2013.01.032}
  {\path{doi:10.1016/j.jcp.2013.01.032}}.

\bibitem{Antoine2005}
X.~Antoine, M.~Darbas, {Alternative integral equations for the iterative
  solution of acoustic scattering problems}, Q. J. Mech. Appl. Math. 58~(1)
  (2005) 107--128.
\newblock \href {http://dx.doi.org/10.1093/qjmamj/hbh023}
  {\path{doi:10.1093/qjmamj/hbh023}}.

\bibitem{Antoine2007}
X.~Antoine, M.~Darbas, {Generalized combined field integral equations for the
  iterative solution of the three-dimensional Helmholtz equation}, ESAIM Math.
  Model. Numer. Anal. 41 (2007) 147--167.
\newblock \href {http://dx.doi.org/10.1051/m2an:2007009}
  {\path{doi:10.1051/m2an:2007009}}.

\bibitem{Darbas2013}
M.~Darbas, E.~Darrigrand, Y.~Lafranche, {Combining analytic preconditioner and
  Fast Multipole Method for the 3-D Helmholtz equation}, J. Comput. Phys. 236
  (2013) 289--316.
\newblock \href {http://dx.doi.org/10.1016/j.jcp.2012.10.059}
  {\path{doi:10.1016/j.jcp.2012.10.059}}.

\bibitem{Kriegsmann1987}
G.~Kriegsmann, A.~Taflove, K.~Umashankar, {A new formulation of electromagnetic
  wave scattering using an on-surface radiation boundary condition approach},
  IEEE Trans. Antennas Prog. 35~(2) (1987) 153--161.
\newblock \href {http://dx.doi.org/10.1109/TAP.1987.1144062}
  {\path{doi:10.1109/TAP.1987.1144062}}.

\bibitem{Jones1988}
D.~Jones, {Surface radiation conditions}, IMA J. Appl. Math. 41 (1988) 21--30.

\bibitem{Jones1990}
D.~Jones, G.~Kriegsmann, {Note on surface radiation conditions}, SIAM J. Appl.
  Math. 50~(2) (1990) 559--568.

\bibitem{Jones1992}
D.~Jones, {An improved surface radiation condition}, IMA J. Appl. Math. 48
  (1992) 163--193.

\bibitem{Ammari1998}
H.~Ammari, S.~He, {An on-surface radiation condition for Maxwell's equations in
  three dimensions}, Microw. Opt. Technol. Lett. 19~(1) (1998) 59--63.

\bibitem{Ammari1998b}
H.~Ammari, {Scattering of waves by thin periodic layers at high frequencies
  using the on-surface radiation condition method}, IMA J. Appl. Math. 60
  (1998) 199--214.

\bibitem{Calvo2004}
D.~Calvo, {A wide-angle on-surface radiation condition applied to scattering by
  spheroids}, J. Acoust. Soc. Am. 116~(3) (2004) 1549--1558.
\newblock \href {http://dx.doi.org/10.1121/1.1777874}
  {\path{doi:10.1121/1.1777874}}.

\bibitem{Calvo2003}
D.~Calvo, M.~Collins, D.~Dacol, {A higher-order on-surface radiation condition
  derived from an analytic representation of a Dirichlet-to-Neumann map}, IEEE
  Trans. Antennas Propag. 51~(7) (2003) 1607--1614.
\newblock \href {http://dx.doi.org/10.1109/TAP.2003.813628}
  {\path{doi:10.1109/TAP.2003.813628}}.

\bibitem{Antoine1999}
X.~Antoine, H.~Barucq, A.~Bendali, {Bayliss-Turkel-like radiation conditions on
  surfaces of arbitrary shape}, J. Math. Anal. Appl. 229 (1999) 184--211.
\newblock \href {http://dx.doi.org/10.1006/jmaa.1998.6153}
  {\path{doi:10.1006/jmaa.1998.6153}}.

\bibitem{Antoine2008}
X.~Antoine, {Advances in the on-surface radiation condition method: Theory,
  numerics and applications}, in: Comput. Methods Acoust. Probl., Saxe-Coburg
  Publications, Stirlingshire, UK, 2008, pp. 169--194.

\bibitem{Antoine2001}
X.~Antoine, {Fast approximate computation of a time-harmonic scattered field
  using the on-surface radiation condition method}, IMA J. Appl. Math. 66~(1)
  (2001) 83--110.
\newblock \href {http://dx.doi.org/10.1093/imamat/66.1.83}
  {\path{doi:10.1093/imamat/66.1.83}}.

\bibitem{Antoine2006}
X.~Antoine, M.~Darbas, Y.~Lu, {An improved surface radiation condition for
  high-frequency acoustic scattering problems}, Comput. Methods Appl. Mech.
  Eng. 195 (2006) 4060--4074.
\newblock \href {http://dx.doi.org/10.1016/j.cma.2005.07.010}
  {\path{doi:10.1016/j.cma.2005.07.010}}.

\bibitem{Barucq2003}
H.~Barucq, {A new family of first-order boundary conditions for the Maxwell
  system: Derivation, well-posedness and long-time behavior}, J. des Math.
  Pures Appl. 82 (2003) 67--88.
\newblock \href {http://dx.doi.org/10.1016/S0021-7824(02)00002-8}
  {\path{doi:10.1016/S0021-7824(02)00002-8}}.

\bibitem{Barucq2010a}
H.~Barucq, R.~Djellouli, A.~Saint-Guirons, {Three-dimensional approximate local
  DtN boundary conditions for prolate spheroid boundaries}, J. Comput. Appl.
  Math. 234 (2010) 1810--1816.
\newblock \href {http://dx.doi.org/10.1016/j.cam.2009.08.032}
  {\path{doi:10.1016/j.cam.2009.08.032}}.

\bibitem{Barucq2012}
H.~Barucq, J.~Diaz, V.~Duprat, {Micro-differential boundary conditions
  modelling the absorption of acoustic waves by 2D arbitrarily-shaped convex
  surfaces}, Commun. Comput. Phys. 11~(2) (2012) 674--690.
\newblock \href {http://dx.doi.org/10.4208/cicp.311209.260411s}
  {\path{doi:10.4208/cicp.311209.260411s}}.

\bibitem{Alzubaidi2016}
H.~Alzubaidi, X.~Antoine, C.~Chniti, {Formulation and accuracy of On-Surface
  Radiation Conditions for acoustic multiple scattering problems}, Appl. Math.
  Comput. 277 (2016) 82--100.
\newblock \href {http://dx.doi.org/10.1016/j.amc.2015.12.023}
  {\path{doi:10.1016/j.amc.2015.12.023}}.

\bibitem{Antoine2004}
X.~Antoine, A.~Bendali, M.~Darbas, {Analytic preconditioners for the electric
  field integral equation}, Int. J. Numer. Methods Eng. 61 (2004) 1310--1331.
\newblock \href {http://dx.doi.org/10.1002/nme.1106}
  {\path{doi:10.1002/nme.1106}}.

\bibitem{Darbas2006}
M.~Darbas, {Generalized combined field integral equations for the iterative
  solution of the three-dimensional Maxwell equations}, Appl. Math. Lett.
  19~(8) (2006) 834--839.
\newblock \href {http://dx.doi.org/10.1016/j.aml.2005.11.005}
  {\path{doi:10.1016/j.aml.2005.11.005}}.

\bibitem{Darbas2015}
M.~Darbas, F.~{Le Louer}, {Well-conditioned boundary integral formulations for
  high-frequency elastic scattering problems in three dimensions}, Math.
  Methods Appl. Sci. 38 (2015) 1705--1733.
\newblock \href {http://dx.doi.org/10.1002/mma.3179}
  {\path{doi:10.1002/mma.3179}}.

\bibitem{Chaillat2015}
S.~Chaillat, M.~Darbas, F.~{Le Louer}, {Approximate local Dirichlet-to-Neumann
  map for three-dimensional time-harmonic elastic waves}, Comput. Methods Appl.
  Mech. Eng. 297 (2015) 62--83.
\newblock \href {http://dx.doi.org/10.1016/j.cma.2015.08.013}
  {\path{doi:10.1016/j.cma.2015.08.013}}.

\bibitem{Atle2007}
A.~Atle, B.~Engquist, {On surface radiation conditions for high-frequency wave
  scattering}, J. Comput. Appl. Math. 204 (2007) 306--316.
\newblock \href {http://dx.doi.org/10.1016/j.cam.2006.02.045}
  {\path{doi:10.1016/j.cam.2006.02.045}}.

\bibitem{Acosta2015f}
S.~Acosta, {On-surface radiation condition for multiple scattering of waves},
  Comput. Methods Appl. Mech. Eng. 283 (2015) 1296--1309.
\newblock \href {http://dx.doi.org/10.1016/j.cma.2014.08.022}
  {\path{doi:10.1016/j.cma.2014.08.022}}.

\bibitem{Stupfel1994}
B.~Stupfel, {Absorbing boundary conditions on arbitrary boundaries for the
  scalar and vector wave equations}, IEEE Trans. Antennas Propag. 42~(6) (1994)
  773--780.
\newblock \href {http://dx.doi.org/10.1109/8.301695}
  {\path{doi:10.1109/8.301695}}.

\bibitem{Murch1993}
R.~D. Murch, {The on-surface radiation condition applied to three-dimensional
  convex objects}, IEEE Trans. Antennas Propag. 41~(5) (1993) 651--658.
\newblock \href {http://dx.doi.org/10.1109/8.222284}
  {\path{doi:10.1109/8.222284}}.

\bibitem{Teymur1996}
M.~Teymur, {A note on higher-order surface radiation conditions}, IMA J. Appl.
  Math. 57 (1996) 137--163.

\bibitem{Yilmaz2007}
B.~Yilmaz, {The performance of the higher-order radiation condition method for
  the penetrable cylinder}, Math. Probl. Eng. 2007 (2007) 17605.
\newblock \href {http://dx.doi.org/10.1155/2007/17605}
  {\path{doi:10.1155/2007/17605}}.

\bibitem{Medvinsky2010}
M.~Medvinsky, E.~Turkel, {On surface radiation conditions for an ellipse}, J.
  Comput. Appl. Math. 234~(6) (2010) 1647--1655.
\newblock \href {http://dx.doi.org/10.1016/j.cam.2009.08.011}
  {\path{doi:10.1016/j.cam.2009.08.011}}.

\bibitem{Chniti2016}
C.~Chniti, S.~Alhazmi, S.~Altoum, M.~Toujani, {DtN and NtD surface radiation
  conditions for two-dimensional acoustic scattering: Formal derivation and
  numerical validation}, Appl. Numer. Math. 101 (2016) 53--70.
\newblock \href {http://dx.doi.org/10.1016/j.apnum.2015.08.013}
  {\path{doi:10.1016/j.apnum.2015.08.013}}.

\bibitem{Nirenberg1973}
L.~Nirenberg, {Pseudodifferential operators and some applications}, in: CBMS
  Reg. Conf. Ser. Math. AMS, Vol.~17, 1973, pp. 19--58.

\bibitem{Chazarain1982-Book}
J.~Chazarain, A.~Piriou, {Introduction to the theory of linear partial
  differential equations}, Studies in Mathematics and its Applications vol. 14,
  Elsevier, 1982.

\bibitem{Taylor2000-Book}
M.~E. Taylor, {Tools for PDE: pseudodifferential operators, paradifferential
  operators, and layer potentials}, Mathematical Surveys and Monographs vol.
  81, AMS, 2000.

\bibitem{Taylor1996-ChapterPSO}
M.~E. Taylor, {Pseudodifferential Operators}, in: Partial Differ. Equations II
  Qual. Stud. Linear Equations, Springer New York, 1996, pp. 1--73.
\newblock \href {http://dx.doi.org/10.1007/978-1-4757-4187-2_1}
  {\path{doi:10.1007/978-1-4757-4187-2_1}}.

\bibitem{DoCarmo1976}
M.~P. {Do Carmo}, {Differential geometry of curves and surfaces}, Prentice-Hall
  Inc., 1976.

\bibitem{Petronetto2013}
F.~Petronetto, A.~Paiva, E.~S. Helou, D.~E. Stewart, L.~G. Nonato, {Mesh-free
  discrete Laplace-Beltrami operator}, Comput. Graph. Forum 32~(6) (2013)
  214--226.
\newblock \href {http://dx.doi.org/10.1111/cgf.12086}
  {\path{doi:10.1111/cgf.12086}}.

\bibitem{Reuter2009}
M.~Reuter, S.~Biasotti, D.~Giorgi, G.~Patan{\`{e}}, M.~Spagnuolo, {Discrete
  Laplace-Beltrami operators for shape analysis and segmentation}, Comput.
  Graph. 33 (2009) 381--390.
\newblock \href {http://dx.doi.org/10.1016/j.cag.2009.03.005}
  {\path{doi:10.1016/j.cag.2009.03.005}}.

\bibitem{Xu2006a}
G.~Xu, {Convergence analysis of a discretization scheme for Gaussian curvature
  over triangular surfaces}, Comput. Aided Geom. Des. 23~(2) (2006) 193--207.
\newblock \href {http://dx.doi.org/10.1016/j.cagd.2005.07.002}
  {\path{doi:10.1016/j.cagd.2005.07.002}}.

\bibitem{Xu2004}
G.~Xu, {Discrete Laplace-Beltrami operators and their convergence}, Comput.
  Aided Geom. Des. 21 (2004) 767--784.
\newblock \href {http://dx.doi.org/10.1016/j.cagd.2004.07.007}
  {\path{doi:10.1016/j.cagd.2004.07.007}}.

\bibitem{Xu2004a}
G.~Xu, {Convergence of discrete Laplace-Beltrami operators over surfaces},
  Comput. Math. Appl. 48 (2004) 347----360.
\newblock \href {http://dx.doi.org/10.1016/j.camwa.2004.05.001}
  {\path{doi:10.1016/j.camwa.2004.05.001}}.

\bibitem{Meyer2003}
M.~Meyer, M.~Desbrun, P.~Schroder, A.~Barr, {Discrete differential-geometry
  operators for triangulated 2-manifolds}, in: Vis. Math. III, Springer Berlin
  Heidelberg, 2003, pp. 35--57.

\bibitem{Surazhsky2003}
T.~Surazhsky, E.~Magid, O.~Soldea, G.~Elber, E.~Rivlin, {A comparison of
  Gaussian and Mean curvatures estimation methods on triangular meshes}, in:
  IEEE Int. Conf. Robot. Autom., 2003, pp. 1021--1026.

\bibitem{Borrelli2003}
V.~Borrelli, F.~Cazals, J.~M. Morvan, {On the angular defect of triangulations
  and the pointwise approximation of curvatures}, Comput. Aided Geom. Des.
  20~(6) (2003) 319--341.
\newblock \href {http://dx.doi.org/10.1016/S0167-8396(03)00077-3}
  {\path{doi:10.1016/S0167-8396(03)00077-3}}.

\bibitem{Magid2007}
E.~Magid, O.~Soldea, E.~Rivlin, {A comparison of Gaussian and mean curvature
  estimation methods on triangular meshes of range image data}, Comput. Vis.
  Image Underst. 107~(3) (2007) 139--159.
\newblock \href {http://dx.doi.org/10.1016/j.cviu.2006.09.007}
  {\path{doi:10.1016/j.cviu.2006.09.007}}.

\bibitem{Bobenko2007}
A.~Bobenko, B.~Springborn, {A discrete Laplace-Beltrami operator for simplicial
  surfaces}, Discret. Comput. Geom. 38 (2007) 740--756.
\newblock \href {http://dx.doi.org/10.1007/s00454-007-9006-1}
  {\path{doi:10.1007/s00454-007-9006-1}}.

\bibitem{Rusinkiewicz2004}
S.~Rusinkiewicz, {Estimating curvature and their derivatives on triangle
  meshes}, in: Symp. 3D Data Process. Vis. Transm., 2004, pp. 486----493.

\bibitem{Grinspun2006}
E.~Grinspun, Y.~Gingold, J.~Reisman, D.~Zorin, {Computing discrete shape
  operators on general meshes}, Comput. Graph. Forum 25~(3) (2006) 547--556.

\bibitem{Sun2016-Book}
X.~Sun, C.~Jiang, J.~Wallner, H.~Pottmann, {Vertex normals and face curvatures
  of triangle meshes}, in: Adv. Discret. Differ. Geom., Springer, 2016, pp.
  267--286.

\bibitem{Arnold2006}
D.~Arnold, R.~Falk, R.~Winther, {Finite element exterior calculus, homological
  techniques, and applications}, Acta Numer. 15 (2006) 1--155.
\newblock \href {http://dx.doi.org/10.1017/S0962492906210018}
  {\path{doi:10.1017/S0962492906210018}}.

\end{thebibliography}

\end{document}